\documentclass[pra,aps,preprint]{revtex4-1}
\usepackage{graphicx}
\usepackage{amsmath}
\usepackage{amsfonts}
\usepackage{amssymb}
\usepackage{mathrsfs}
\usepackage{amsthm}
\theoremstyle{definition}
\theoremstyle{definition}

\newtheorem{theorem}{Theorem}[]
\newtheorem{corollary}[theorem]{Corollary}
\newtheorem*{definition}{Definition}

\begin{document}

%
%

\title{Some results on topological currents in field theory}

\author{Vivek M. Vyas}

\address{Raman Research Institute, \\C. V. Raman Avenue, Sadashivnagar, \\Bengaluru 560 080, INDIA\\
vivekv@rri.res.in}

\author{V. Srinivasan}

\address{Department of Theoretical Physics, Guindy Campus,\\ University of Madras, Chennai 600 025, INDIA}

\author{Prasanta K. Panigrahi}

\address{Indian Institute of Science Education and Research Kolkata, \\Mohanpur, Nadia 741 246, INDIA}

\begin{abstract}
	A few exact results concerning topological currents in field theories are obtained. It is generally shown that, a topological charge can not generate any kind of symmetry transformation on fields. It is also proven that, the existence of a charge that does not generate any kind of symmetry transformation on fields, has to be of topological origin. As a consequence, it is found that in a given theory, superconductivity via Anderson-Higgs route can only occur if the gauge coupling with other fields is minimal. Several physical implications of these results are studied.
\end{abstract}

\maketitle



\section{Introduction}

Topological aspects are central to several problems in quantum field theory. In study of problems related to monopoles, instantons, anomalies and $\sigma$-models, ideas of topology are indispensable \cite{nair,jackiw,schw}. Insights based on topology are useful in the understanding of various phenomena in areas like superconductivity, Hall effect, QCD and cosmology \cite{nair,ume2,tsv}. It has been generally seen that topology often governs the global aspects in a given field theory, and provides with information which is non-perturbative in nature.               

In this paper, the aspects of currents which are of topological origin, and are conserved without resorting to any symmetry principle are studied. It is generally proven that, topological charges do not generate any symmetry transformation on the dynamical fields of the theory. Conversely, it is also proven that, if there exists a charge that does not generate any transformation on the dynamical fields of the theory, then the corresponding current must be topological in nature. One of the interesting implication of these statements is that, topological spontaneous symmetry breaking is non-existent. As a result, it is found that in a given theory, superconductivity via Anderson-Higgs route can only occur if the gauge coupling is minimal. This has several consequences in condensed matter and nuclear physics. Occurrence of topological currents in certain non-relativistic models are also studied and their physical relevance pointed out.       

\section{Topological currents}\label{tc}

As is well known, the dynamics of a given (classical) field theory of a field $\phi(x)$, defined in $n+1$-dimensional space-time, is governed by the equation of motion, which follows from extremisation of the action functional, ${S} = \int dV_{x} dt \: \mathscr{L} (x,\phi(x), \partial \phi(x))$\footnote{Throughout this paper we shall denote n-dimensional spatial volume element by $dV_{x}=dx_{1}dx_{2} \cdots dx_{n}$. Sometimes for brevity subscript may be dropped in $dV_{x}$. Our metric convention is $\eta_{\mu \nu} = diag(1,-1,-1, \cdots, -1)$. The notations and conventions of that of Itzykson and Zuber \cite{iz} are followed, unless mentioned otherwise.} as per the Hamilton's principle. Given such a theory, classical Noether theorem \cite{kosmann,brown,olver} dictates that, if continuous transformation of coordinates ($x^{\mu} \rightarrow x^{\mu} + \delta x^{\mu} $), fields ($\phi(x) \rightarrow \phi(x) + \delta \phi(x) $) and possibly Lagrangian (density) ($\mathscr{L} \rightarrow \mathscr{L} + \delta \mathscr{L} $), is such that action transforms as: ${S} \rightarrow {S} + \int dV dt \: d_{\mu} \Lambda^{\mu}$, then correspondingly there exists a current $j^{\mu}$ in the theory, such that $d_{\mu} j^{\mu} = 0$\footnote{Here $\frac{d}{dx^{\mu}} = d_{\mu}$ stands for total derivative with respect to $x^{\mu}$, which takes into account both implicit and explicit $x^{\mu}$ dependence, unlike $\partial_{\mu}$ which is only concerned with explicit $x^{\mu}$ dependence \cite{brown,rosen}. For example: $d_{\mu} f(x) = \partial_{\mu} f(x)$ and $d_{\mu} f(x, g(x)) = \partial_{\mu} f(x,g(x)) + \partial_{\mu}g(x) \frac{\partial}{\partial g(x)} f(x,g(x))$.}. Further, if the current $\vec{j}(x)$ decays sufficiently fast as $|\vec{x}| \rightarrow \infty$ then charge $Q = \int dV j^{0}(x)$ is conserved under time evolution $\frac{d Q}{dt} = 0$\footnote{The statement that, as $|\vec{x}| \rightarrow \infty$, current $\vec{j}(x)$ decays sufficiently fast, in strict sense, is a statement about the kind of boundary conditions the system is assumed to obey.}. Assuming that the passage from this Lagrangian configuration space description to Hamiltonian phase space description is non-pathological, charge conservation now can be stated as: $ \frac{\partial Q}{\partial t} + \left\lbrace Q, H \right\rbrace$ = 0, where the second term stands for Poisson bracket of $Q$ and Hamiltonian $H$ of the theory. In this description, one finds that charge $Q$ indeed generates transformations $\delta F(\phi,\pi) = \epsilon \left\lbrace F(\phi,\pi), Q \right\rbrace$, on any function F of field $\phi$ and its canonical momentum $\pi$, which are precisely the ones that give rise to the conserved current: $d_{\mu} j^{\mu} = 0$ via Noether theorem \cite{schwinger}. It must be noted that current $j^{\mu}(x)$ is a function of dynamical fields $\phi(x)$ and $\pi(x)$, and the statement $d_{\mu} j^{\mu} = 0$ holds true only for those field configurations which solve equation(s) of motion \cite{olver}. Such statements are often referred to in the literature as weak statements \cite{rosen} and we shall denote them as $d_{\mu} j^{\mu} \overset{w}= 0$. 

A simple and instructive example of this appears in the theory of free real scalar field $\phi(x)$ defined in 1+1 dimensional spacetime with action $S = \int dxdt \: \frac{1}{2} \partial_{\mu} \phi \partial^{\mu} \phi$. Canonical conjugate of $\phi(x)$ is $\pi=\frac{\partial \phi}{\partial t}$, which obeys canonical equal time Poisson brackets: $\left\lbrace \phi(x,t), \pi(y,t) \right\rbrace = \delta(x-y)$. The equation of motion for $\phi$ field is: $\partial^{2} \phi = 0$.
As is evident, above action is invariant under continuous field transformation: $\phi \rightarrow \phi + \text{constant}$, and so as per Noether theorem one obtains a current $j^{\mu} = \partial^{\mu} \phi$, which is conserved $\partial_{\mu} j^{\mu} = \partial_{\mu}\partial^{\mu} \phi = 0$. Note that the current conservation equation itself coincides with equation of motion and so the current conservation only holds as long as $\phi(x)$ obeys equation of motion. Hence, the statement of current conservation is a weak one: $\partial_{\mu} j^{\mu} \overset{w}=0$. Also note that conserved charge $Q=\int dx \: \pi(x)$, generates field transformation $\delta \phi(x) = \epsilon \left \lbrace \phi(x), Q \right \rbrace = \epsilon$, where $\epsilon$ is a real infinitesimal.

Unlike the above mentioned Noether currents, there exist another class of currents called topological currents, which are conserved due to topological reasons. Their conservation is not in reference to any particular action or Hamiltonian under consideration, and holds identically in general\footnote{Olver  \cite{olver} refers to such conservation laws as null divergences, and proves a theorem which shows their connection with Poincar\'{e} lemma. {For discussions on the role of topology in this context, readers are referred to Ref. \cite{nair} and Ref. \cite{patani}.}}. An example of such a current, in case of real scalar field in $1+1$ dimensional space-time is: $j_{T}^{\mu} = \epsilon^{\mu \nu} \partial_{\nu} \phi$. It is clear that the conservation of this current is not dependent on any particular form of action, and neither is $\phi(x)$ required to obey the equation of motion. Also note that the conservation is insensitive to space-time geometry as well, since $\partial_{\mu} j_{T}^{\mu} = 0$ is independent of metric \footnote{Topological currents can be studied using the coordinate free language of differential forms, as is done in detail in Ref. \cite{patani}. Such a study brings to fore the global topological aspects of these currents.}. 

In what follows, we shall work with following definition of topological current:
\begin{definition}
	In a given field theory, a \emph{topological current} $j_{T}^{\mu}$, which is in general a function of coordinates, dynamical fields and possibly of the derivatives of dynamical fields, is a conserved current, whose conservation holds identically: $d_{\mu} j_{T}^{\mu} \overset{s}= 0$.	
\end{definition}

Such a topological conservation, hence holds in the strong sense, in contrast to weak sense defined above; and holds true for all the dynamical field configurations which in general \emph{do not} solve equations of motion of the theory\footnote{Topological currents are intimately connected with local gauge invariances, which was shown by E. Noether, (not so wellknown \emph{Noether Second Theorem}) while working within the Lagrangian framework \cite{brading}. In this paper, the discussion is confined within Hamiltonian framework, which does not admit local invariances. Nevertheless it would be interesting to explore the implications of the Noether second theorem in canonical framework. Recently its implications in Weyl invariant theories have been explored in Ref. \cite{jack} and Ref. \cite{shukla}}.

It must be understood however that, the classification of a given current as Noether or topological is \emph{not} an absolute one, and depends on the dynamical fields used to express the current. This can be clearly seen from the real scalar theory example considered above. The conserved current $j^{\mu} = \partial^{\mu} \phi$ is a Noether current when $\phi$ is treated as a dynamical field, with an equation of motion $\partial^2 \phi = 0$. Now if one defines a new field $\chi$, such that $\partial^{\mu} \phi = \epsilon^{\mu \nu} \partial_{\nu} \chi$, or explicitly $\chi(x,t) = \int_{-\infty}^{x} dy \; \partial_{t}{\phi(y,t)}$, one sees that the same current $j^{\mu}$ now becomes topological $j^{\mu} = \partial^{\mu} \phi = \epsilon^{\mu \nu} \partial_{\nu} \chi$. Thus one learns that \emph{any statement pertaining to currents, in particular, about their being Noether or topological, makes sense only with reference to a given set of dynamical fields}.

As seen above, the conservation of these topological currents is not in reference to any particular equation of motion, and so it is easy to see that, they are not connected to any continuous field transformation. This fact can be expressed as:

\begin{theorem} \label{t1} Assume that there exists a classical field theory of a dynamical bosonic field $\phi(x)$, which is Poincar\'{e} invariant. Let $\pi(y)$ denote canonical momentum corresponding to $\phi(x)$, obeying canonical Poisson bracket $\left\lbrace \phi ( \vec{x},t),\pi(\vec{y},t) \right\rbrace = \delta (\vec{x}-\vec{y})$. If there exists a conserved topological current $J^{\mu}_{T}$ in such a theory, then the corresponding conserved charge $Q_{T} = \int dV J^{0}_{T}$ does not generate any transformations: $\left\lbrace Q_{T}, \phi(x) \right\rbrace = 0$ and $\left\lbrace Q_{T}, \pi(x) \right\rbrace = 0$.   
\end{theorem}

\begin{proof}
	From the definition of topological current, it immediately follows that the current $J^{\mu}_{T}$ is strongly conserved: $d_{\mu} J_{T}^{\mu} \overset{s}= 0$. Equation of motion for any observable $O(x,\phi(x),\pi(x))$ can be written as \cite{taka}: 
	\begin{equation} \nonumber
	d_{\mu} O = \partial_{\mu} O + \left\lbrace O, P_{\mu}\right\rbrace, 
	\end{equation}
	where $P_{\mu}$ are generators of space-time translation, with $P_{0}=H$, the Hamiltonian of the theory, and $\vec{P} = \int dV \: \pi \vec{\nabla} \phi$. By virtue of above equation of motion, one can write current conservation as:
	\begin{equation} \label{e1}
	d_{\mu} J_{T}^{\mu} = \partial_{\mu} J_{T}^{\mu} + \left\lbrace J_{T}^{\mu}, P_{\mu}\right\rbrace = 0.
	\end{equation}
	But since topological current is strongly conserved, the conservation holds even if the theory is defined with modified  Hamiltonian $\tilde{H}$ (which is assumed to be such that the modified theory is also Poincar\'{e} invariant): 
	\begin{equation} \label{e2}
	d_{\mu} J_{T}^{\mu} = \partial_{\mu} J_{T}^{\mu} + \left\lbrace J_{T}^{\mu}, \tilde{P}_{\mu}\right\rbrace = 0.
	\end{equation}
	This immediately implies: $\left\lbrace H - \tilde{H} , J_{T}^{0} \right\rbrace = 0$. Considering the case when $H - \tilde{H} = \int dV_{y} \: F(\phi(y),\pi(y))$, where $F$ is some polynomial function\footnote{One can consider $F$ to be any general local function of dynamical fields, here polynomial function is considered for the sake of simplicity.} of $\phi$ and $\pi$; one finds that: $\int dV_{y} \left\lbrace F(\vec{y},t), J_{T}^{0}(\vec{x},t) \right\rbrace = 0$. Owing to translational invariance one has: $\left\lbrace F(\vec{y},t), J_{T}^{0}(\vec{x},t) \right\rbrace = \mathscr{F}(\vec{y} - \vec{x},t)$, where $\int dV_{y} \: \mathscr{F}(\vec{y} - \vec{x},t) = 0$. Exploiting the translational invariance of measure, this can be written as: $\int dV_{y} \: \mathscr{F}(\vec{y},t) = 0 = \int dV_{y} \: \mathscr{F}(\vec{z}+\vec{y},t)$, where $\vec{z}$ is an independent constant vector. This can be rewritten as:
	\begin{equation} \nonumber
	\int dV_{y} \: \left\lbrace F(\vec{z},t), J_{T}^{0}(\vec{y},t) \right\rbrace = 0,
	\end{equation}
	or $\left\lbrace F(\vec{z},t), Q_{T} \right\rbrace = 0$. This can not be true for a polynomial function $F(\phi,\pi)$, unless $\left\lbrace \phi(\vec{x},t), Q_{T} \right\rbrace = 0$ and $\left\lbrace \pi(\vec{x},t), Q_{T} \right\rbrace = 0$. 
\end{proof}

It can be easily seen that a similar statement will also hold true for the case of quantised theory, that is, \emph{commutator of topological charge operator $Q_{T}$ with dynamical field operators is always vanishing: $\left[Q_{T}, \phi(x) \right] = 0$, $\left[Q_{T}, \pi(x) \right] = 0$}. This fact was observed by Ezawa \cite{ezawa} in context of a particular model, leading to topological superselection rules in such theories \cite{wick,wightman}. Later Shamir and Park offered a general proof for this statement \cite{shamir}; while working in a relativistic quantum field theoretic framework. Also note that, although above theorem was proved for bosonic theories, the same will also hold for fermionic theory, with Poisson brackets being replaced by generalised Poisson brackets \cite{hen}. The above theorem was proved for theories that possess Poincar\'{e} invariance, however, it is observed that the same also holds for non-relativistic theories as well.

It is worth mentioning that, even though $\left[Q_{T}, \phi(x) \right] = 0$, it is possible to construct nonlocal functions of $\phi$ such that 
$\left[Q_{T}, f(\phi(x)) \right] \neq 0$, without conflicting with the definition of topological charge. As was pointed out by Nakanishi \cite{naka,ojima}, nonvanishing of such commutators occurs because of the nonlocal nature of the function $f(\phi)$. This can be clearly seen in the example of two dimensional real scalar theory: $\partial^2 \phi = 0$. As seen earlier, there exists a topological charge $Q$ in the theory: 
$Q_{T} = \int dx \; \phi_x$, which arises from the conserved current $j^{\mu} = \epsilon^{\mu \nu} \partial_{\nu} \phi$. From the commutation relations for $\phi$ it is clear that $\left[Q_{T}, f(\phi(x)) \right] = 0$ for any local function $f$ of $\phi$. The current $j^{\mu}$ can also be expressed in terms of another field $\chi$ such that: $\partial^{\mu} \chi = \epsilon^{\mu \nu} \partial_{\nu} \phi = j^{\mu}$, so that $\partial^{2} \chi = 0$ is the equation of motion. From the commutation relation for $\chi$, one immediately sees that: $\left[Q, \chi(x) \right] \neq 0$, which means that $Q$ is a Noether charge when the theory is described in terms of $\chi$. This statement means that: $\left[Q, \int_{-\infty}^{x} dy \; \partial_{t}{\phi(y,t)} \right] \neq 0$, without conflicting with $\left[Q, \phi(x) \right] = 0$. Because of the ill-defined nature of commutator owing to the nonlocal nature of function, one sees that $\left[Q, \int_{-\infty}^{x} dy \; \partial_{t}{\phi(y,t)} \right] \neq \int_{-\infty}^{x} dy \; \partial_{t} \left[Q, {\phi(y,t)} \right]$\footnote{One notes an interesting parallel of this with that of differentiation under integral sign:
	$\frac{d}{dt} \int_{a(t)}^{b(t)} dx \; f(x,t) = \int_{a(t)}^{b(t)} dx \; \frac{\partial}{\partial t} f(x,t) + f(b(t),t) b'(t) - f(a(t),t) a'(t)$, by identifying differentiation with derivation: $\frac{\partial}{\partial t} (\cdots) \leftrightarrow  [Q, (\cdots)]$.}. Interestingly, the physical observables in a field theory are usually required to be local functions of the dynamical fields, and hence one is usually spared working with such pathological commutators. A prudent way of working with such commutators is by treating such nonlocal fields as dynamical, and expressing currents and charges (and all the observables) in terms of these fields. It is easy to see that when the theory is expressed in terms of these new fields, charge $Q$ would not be topological any more. This is very clearly seen in above example as well as in the equivalence between sine-Gordon and Thirring model and in case of $\sigma$-model \cite{mandel,ojima,bar}, topological currents in one description become Noether in the other and vice versa.

Interestingly, one notes that $\left\lbrace H, Q_{T} \right\rbrace = 0$, for any Hamiltonian $H$. From Hamilton's equation (Heisenberg equation in quantum case) one finds that $\frac{\partial Q_{T}}{\partial t} = 0$, which states that \emph{topological charges can not have any explicit time dependence}.

In above discussion, we have assumed that the field theory is defined with suitable boundary conditions, which are such that, the conserved space-time charges $(P_{\mu})$ and topological charge $Q_{T}$ exist and are well defined. The same will also be assumed in what follows.

Although above theorem was proven assuming that field $\phi(x)$ is a single component object, it is straightforward to see that it will also hold, when field $\phi(x)$ in general is a multi-component field. Apart from $1+1$ dimensional real scalar theory discussed above, another example of this theorem is given by the nonrelativistic theory of bosons, living on a line, governed by complex field $\psi(x,t)$. Such a theory admits three conserved topological currents: (a) $j_{0} = \partial_{x} (\psi + \psi^{\dagger}), j_{x} = - \partial_{t}(\psi + \psi^{\dagger})$; (b) $j_{0} = -i \partial_{x} (\psi - \psi^{\dagger}), j_{x} = i \partial_{t}(\psi - \psi^{\dagger})$ and (c) $j_{0} = \partial_{x} (\psi^{\dagger} \psi), j_{x}=- \partial_{t}(\psi^{\dagger} \psi)$. Corresponding conserved charges are given by $Q_{1} = \int dx \: \partial_{x} (\psi + \psi^{\dagger})$, $Q_{2} = \int dx \: i\partial_{x} (\psi^{\dagger}-\psi)$ and $Q_{3} = \int dx \: \partial_{x} (\psi^{\dagger} \psi)$. Owing to canonical commutation relations: $[\psi(x,t),\psi^{\dagger}(y,t)]= \delta(\vec{x}-\vec{y})$, one sees that $[Q_{(1,2,3)}, \psi(x,t)] = 0$, irrespective of the form of Hamiltonian.  

Yet another example of above theorem, in case of fermions, is given by quantised field theory of spinor field $\Psi(x)$; where the conserved current, $J_{t}^{\mu} = \partial_{\nu} (\bar{\Psi} \sigma^{\mu \nu} \Psi)$, ($\sigma^{\mu \nu} = \frac{i}{2} \left[ \gamma^{\mu}, \gamma^{\nu} \right]$) is a topological current. The topological charge $Q=\int dV \partial_{\nu} (\bar{\Psi} \sigma^{0 \nu} \Psi)$ commutes with both $\Psi(x)$ and $\Psi^{\dagger}(x)$, owing to equal-time anticommutation relation: $\left[ \Psi(\vec{x},t),\Psi^{\dagger}(\vec{y},t) \right]_{+} = \delta(\vec{x} - \vec{y})$.    
In this light, it is very interesting to note that, the conserved charge current $J^{\mu} = \bar{\Psi} \gamma^{\mu} \Psi$, which is the Noether current corresponding to U(1) phase invariance $\Psi \rightarrow \Psi \: e^{i \theta}$ (where $\theta$ is a real constant), can be written using Gordon decomposition \cite{sakurai,hehl} as: 
\begin{equation} \nonumber
J^{\mu} = \frac{i}{2m} \left( \bar{\Psi} \partial^{\mu} \Psi - \partial^{\mu} \bar{\Psi} \Psi \right) + \frac{1}{2m} \partial_{\nu} \left( \bar{\Psi} \sigma^{\mu \nu} \Psi \right).
\end{equation}   
The first term in above expression is the orbital current or the diamagnetic term, whereas the second term is the spin term or paramagnetic term, which is topological. Correspondingly, Noether charge $Q=\int dV \: J^{0}$ has two components, first one corresponding to orbital current: $Q_{o}=\int dV \: \frac{i}{2m} \left( \bar{\Psi} \partial^{0} \Psi - \partial^{0} \bar{\Psi} \Psi \right)$ and second one corresponding to topological current: $Q_{t} = \int dV \: \frac{1}{2m} \partial_{\nu} \left( \bar{\Psi} \sigma^{0 \nu} \Psi \right)$. From this decomposition it is clear (as per Theorem \ref{t1}) that, the U(1) symmetry transformation is actually generated by $Q_{o}$, since $Q_{s}$ can not give rise to such transformation. This is an important result, which tells us that, although Noether theorem provides us with an expression for the conserved charge, such a conserved charge may contain in general a contribution due to topological charges. As per Theorem \ref{t1} such a contribution is harmless as far as symmetry properties are concerned.


It is worth mentioning that the above theorem also holds for gauge theories, which can be elegantly seen in case of Abelian gauge field $A_{\mu}$ defined in $2+1$ dimensional spacetime. 
Such a theory admits a topological current: $J^{\mu}_{T} = \epsilon^{\mu \nu \rho} F_{\nu \rho}$, with topological charge $Q_{T} = \int d^{2}x \: \epsilon^{ij}F_{ij}$ (here Latin indicies take values $1,2$). Note that conservation of topological charge is independent of gauge choice. In covariant gauge, the equal time commutator is $[A_{\mu}(x), F_{ij}(y)]_{x_{0}=y_{0}} = 0$, hence one has $[Q_{T}, A_{\mu}]=0$. Similarly in axial gauge ($A_{0}=0$), one has $[A_{i}(x), F_{jk}(y)]_{x_{0}=y_{0}} = 0$ and hence $[Q_{T}, A_{i}]=0$.         

An important implication of Theorem \ref{t1} is regarding Poincar\'{e} invariance of the theory. 
It is a well known fact that canonical energy-momentum tensor $\Theta^{\mu \nu}$ in a relativistic quantum field theory, obtained using Noether theorem is conserved $\partial_{\mu} \Theta^{\mu \nu} = 0$ and is in general not symmetric $\Theta^{\mu \nu} \neq \Theta^{\nu \mu}$. It was shown by Belinfante that canonical energy-momentum tensor can always be improved by adding a specifically constructed term $\partial_{\rho} \chi^{\rho \mu \nu}$ (where $\chi^{\rho \mu \nu} = - \chi^{\mu \rho \nu}$), such that improved energy-momentum tensor $T^{\mu \nu}=\Theta^{\mu \nu} + \partial_{\rho} \chi^{\rho \mu \nu}$, is symmetric $T^{\mu \nu} = T^{\nu \mu}$ \cite{nair,greiner}. By construction $\partial_{\mu} ( \partial_{\rho }\chi^{\rho \mu \nu} ) = 0$ identically, which ensures conservation of the improved energy-momentum tensor $\partial_{\mu} T^{\mu \nu} = 0$. This results in modification of canonical angular momentum tensor $M^{\mu \nu \rho}$ to improved angular momentum tensor $L^{\mu \nu \rho} = M^{\mu \nu \rho} + \partial_{\sigma} \eta^{\sigma \mu \nu \rho}$, where $\eta$ is antisymmetric in first two indicies, and can be written in terms of $\chi$ \footnote{Redefinition of energy-momentum tensor and angular momemtum tensor is also discussed in detail in Ref. \cite{hehl}}. In a given theory, if the charges constructed out of canonical energy-momentum and angular momentum tensors, obey Poincar\'{e} algebra, then the theory is said to be Poincar\'{e} invariant. Now consider a given theory, where the charges constructed out of canonical energy-momentum and angular momentum tensors, obey Poincar\'{e} algebra. However, when one constructs charges out of improved energy-momentum and angular momentum tensors, one wonders whether these improved charges are guaranteed to obey Poincar\'{e} algebra in general or not. As mentioned above, both the improvements, respectively to energy-momentum and angular momentum tensor, are topological currents. The charges constructed out of improved tensors will differ from the canonical one by a contribution coming from the topological part, which by Theorem \ref{t1}, would not contribute to any commutation relations, and hence can not spoil Poincar\'{e} invariance of the theory. It is worth emphasising that, this conclusion also holds in general, that is, \emph{under any redefinition of a conserved Noether current by addition of a topological current, the commutation relations amongst Noether charges remain unaffected}. 

Above assertion also implies that, if in a theory the energy-momentum tensor itself is conserved topologically, then such a theory is non-dynamical. Chern-Simons theory defined in $2+1$ dimensional space-time with action: $\mathscr{L} = \epsilon^{\mu \nu \rho} A_{\mu} F_{\nu \rho}$, is an instructive example of this \cite{deser}. For this theory, the canonical energy-momentum tensor $\theta_{\mu \nu}$ vanishes identically and hence is conserved topologically (as per our definition). The improved energy-momentum tensor, in general, will be of the form $T^{\mu \nu} = \epsilon^{\mu \nu \rho} \partial_{\rho} \phi$ (where $\phi(x)$ is some function whose exact form is irrelevant in this discussion), which is also a topologically conserved current. As is evident, both the energy-momentum tensors lead to same conclusion that, the theory is non-dynamical.                        

From above discussion it is clear that, the complete set of conserved currents can be divided into exactly two equivalence classes: Noether currents and non-Noether or topological currents. Note that for any conserved charge $Q_{i}$ in the theory, charge $\lambda Q_{i}$ is also conserved, for each $\lambda \in \mathbb{R}$. Further, a linear combination of any two conserved charges $Q_{i,j}$, $\lambda_{1} Q_{i} + \lambda_{2} Q_{j}$ is also a conserved charge, for all $\lambda_{1,2} \in \mathbb{R}$. From these two properties (which are vector addition and scalar multiplication), it easy to see that, the set of all conserved charges, for a given theory, form a linear vector space $\mathbf{V}$ under addition, and over field $\mathbb{R}$. Interestingly, since addition of two topological charges is also a topological charge, and zero charge (which is the identity element) is topological, one sees that the set of conserved topological charges itself forms a vector space $\mathbf{T}$, which is a subspace of $\mathbf{V}$. Note that, all the Noether charges live in $\mathbf{T}^{c}$, which is the complement of $\mathbf{T}$ and a subset of $\mathbf{V}$. Further, any two Noether charges which differ by a scalar multiple or an element in $\mathbf{T}$ generate same transformation. So it is imperative that, one identifies all elements in $\mathbf{V}$ which differ from each other by an element in $\mathbf{T}$, defining a quotient space $\mathbf{V}/\mathbf{T}$ \cite{halmos}. Now one can use the fact that, quotient space $\mathbf{V}/\mathbf{T}$ forms a linear vector space under addition with field being $\mathbb{R}$, and construct a suitable complete basis set. Elements of such a basis set, each being a linearly independent Noether charge, will generate a distinct transformation.   

In light of above discussions, it is worth enquiring whether converse of Theorem \ref{t1} is true or not. That is, for a given charge  (which can be written as surface integral of a current, with integration being done over a space-like surface) that commutes with all the dynamical fields of the theory, does it imply that the corresponding current is strongly conserved ? Below it is shown that the answer to this question is affirmative. 

\begin{theorem} \label{t2} Assume that there exists a classical field theory of a bosonic dynamical field $\phi(x)$, which is Poincar\'{e} invariant. Let $\pi(y)$ be the canonical momentum corresponding to $\phi(x)$, obeying canonical Poisson bracket: $\left\lbrace \phi ( \vec{x},t),\pi(\vec{y},t) \right\rbrace = \delta (\vec{x}-\vec{y})$. If there exists a four-vector current $j^{\mu}$ in the theory, such that the corresponding conserved charge $Q = \int_{\sigma} ds_{\mu} j^{\mu}$ {\footnote{Here $\sigma$ denotes the space-like surface over which the integration is to be done \cite{schwinger}, generally taken to be a constant time surface. Here $ds_{\mu}$ stands for oriented area element on the surface $\sigma$.}} does not generate any transformations: $\left\lbrace Q_{T}, \phi(x) \right\rbrace = 0$ and $\left\lbrace Q_{T}, \pi(x) \right\rbrace = 0$, then the corresponding current is identically conserved: $\partial_{\mu} j^{\mu} \overset{s} = 0$.
\end{theorem}

\begin{proof} Since $Q$ commutes (in Poisson bracket sense) with both $\phi(x)$ and $\pi(x)$, it implies that $\left\lbrace Q_{T}, f(y) \right\rbrace = 0$, where $f(y)$ is some (local) function of $\phi(y)$, $\pi(y)$ and possibly their derivatives, with $y$ being a fixed coordinate. This implies:
	\begin{equation}
	\int_{\sigma} ds_{\mu} \left\lbrace j^{\mu}(x), f(y) \right\rbrace = 0.
	\end{equation} 
	Defining $B_{\mu}(x-y)=\left\lbrace j_{\mu}(x), f(y) \right\rbrace$, this can be written as:
	\begin{equation}\label{div}
	\int_{\sigma} ds_{\mu} B^{\mu}(x-y) = 0.
	\end{equation}
	Without loss of generality we set $y=0$ in what follows. It is reasonable to assume that, the boundary conditions obeyed by the fields are such that, as $|\vec{x}| \rightarrow \infty$ the dynamical fields go to a constant (usually zero) \emph{i.e.,} $\phi(|\vec{x}| \rightarrow \infty) \rightarrow c$ and similarly $\pi(|\vec{x}| \rightarrow \infty) \rightarrow c'$, where $c$ and $c'$ are constants. The same can be stated in a covariant manner: $\phi({x}_{\mu} x^{\mu} \rightarrow - \infty) \rightarrow c$ and $\pi({x}_{\mu} x^{\mu} \rightarrow - \infty) \rightarrow c'$. As a result of these conditions, functions of dynamical fields, like $B_{\mu}(x)$, attain a constant value as ${x}_{\mu} x^{\mu} \rightarrow - \infty$. It is beneficial to work in Euclidean space-time \cite{jackiwd}, by performing Wick rotation on all the vectors: $A_{0} \rightarrow i A_{0}, \vec{A} \rightarrow \vec{A}$. Above boundary condition, now becomes $B_{\mu} \rightarrow constant$, as ${x}_{\mu} x_{\mu} \rightarrow  \infty$. This allows us to identify space-time infinity as a single point for such functions, which amounts to say that such functions are actually living on a space-time which is 
	topologically a sphere\footnote{Mathematically this is one point compactification \cite{munkres} of Euclidean space-time $\mathbb{R}^{n}$ to n-sphere $\mathbb{S}^{n}$.}. With this identification, above expression (\ref{div}) actually becomes a surface integral over a closed space-like surface\footnote{This closed space-like surface is also a sphere, albeit of one lesser dimension than the space-time sphere.} on the space-time sphere:
	\begin{equation}
	\oint_{\sigma} ds_{\mu} B_{\mu}(\vec{x},t) = 0.
	\end{equation} 
	Using the divergence theorem this can be written as a volume integral over the space-time region enclosed by the surface: 
	\begin{equation}
	\int dV \: \partial_{\mu} B_{\mu} = 0.
	\end{equation}
	Owing to the generality in the definition of $B_{\mu}$, above statement can only hold, if $\partial_{\mu} B_{\mu}(\vec{x},t) = 0$, which in turn implies $\partial_{\mu} j_{\mu}(\vec{x},t) = \text{constant}$. Since $Q$ is a conserved charge, consistency requires that $\partial_{\mu} j_{\mu}(\vec{x},t) = \text{constant} = 0$.  Wick rotating to Minkowski space-time, one finds that $\partial_{\mu} j^{\mu}(x) = 0$. Note that in arriving at this result, we have not assumed at any juncture that the dynamical fields which constitute $j_{\mu}(x)$ solve equation of motion, and so current conservation is a strong equality and holds identically: $\partial_{\mu} j^{\mu} \overset{s} = 0$.    
\end{proof}
It can be easily seen that, this theorem will also hold for fermionic theories as well.

{In this paper the discussion is confined only to those conserved topological charges $Q_T$ which arise as a result of a conserved topological current $\partial_{\mu} j^{\mu} \overset{s} = 0$. However, in general, it is possible that a theory may admit conserved topological charge $Q_\tau$, which would commute with all the dynamical fields in the theory, and yet does not arise as a result of any conserved current. Such conserved quantities were considered long back, in 1952 in a remarkable paper by Wick, Wightman and Wigner \cite{wick}, where it was shown that existence of topological charges would give rise to \emph{superselection sectors} in the theory \cite{wightman}. Infact it was found that the transformation properties of fields under discrete space-time symmetries, like parity and time-reversal, would give rise to superselection rules in the theory, based on which the notion of intrinsic parity of a particle was proposed \cite{wick,wightman}. One can infact investigate both topological currents and corresponding conserved charges from this point of view, particularly in light of their behaviour under discrete symmetries like parity. Presumably such a discussion would shed light on various superselection sectors in the theory. In this paper, however we shall not dwell upon these interesting aspects of topological charges, and shall report some of our ongoing work in this direction elsewhere in due course.} 

\section{Coupling with gauge field}

It is well known that in a quantum field theory, in constrast to quantum mechanics, vacuum in general is not unique \cite{ume}. The celebrated Goldstone theorem \cite{goldstone} for a given quantised field theory of field $\phi(x)$ states that, if for a conserved charge $Q$, the vacuum expectation value of Goldstone commutator $[Q, \phi(x)]$ is nonvanishing $\langle vac | [Q, \phi(x)] |vac \rangle \neq 0$, then the symmetry corresponding to charge Q is said to be spontaneously broken in that vacuum, implying existence of gapless Nambu-Goldstone modes in field $\phi(x)$. The same in general holds true for nonvanishing of vacuum expectation of Goldstone commutator $\langle vac | [Q, F(x)] |vac \rangle$ for any composite field $F(\phi(x),\Pi(x))$, as in the case of superconductivity \cite{ume}. As shown by Nambu \cite{nambu}, the symmetry associated with charge $Q$ is never lost, but is rearranged and gets manifested in such a vacuum by presence of Nambu-Goldstone modes \cite{ume}.

In a given theory, if the symmetry generated by charge $Q$ is spontaneously broken, and if a gauge field is coupled to the conserved charge $Q$, then via Anderson-Higgs mechanism, the gauge field gets gauge invariant mass \cite{ume,ojima}. In case of superconductivity, it is well known that $Q$ corresponds to electric charge, and the electromagnetic field becomes massive \cite{ume}. Now if charge $Q$ happens to be a topological charge, then from Theorem \ref{t1}, one arrives at an interesting result: 

\begin{theorem} \label{t3}
	In a given quantised field theory of some field $\phi(x)$, obeying appropriate canonical (anti)commutation relation: $\left[ \phi ( \vec{x},t),\Pi(\vec{y},t) \right]_{\pm} = \mp i \delta (\vec{x}-\vec{y})$, if there exists a conserved topological charge $Q_{T}$, such that the gauge field $A_{\mu}$ only couples to $Q_{T}$, then gauge field $A_{\mu}$ can not become massive via Anderson-Higgs mechanism in such a theory.    
\end{theorem}

\begin{proof}
	The proof is straightforward. Note that for topological charge $Q_{T}$, Goldstone commutator vanishes for both $\phi$ and $\Pi$: $[Q_{T}, \phi(x)]= 0$ and $[Q_{T}, \Pi(x)] = 0$ due to Theorem \ref{t1}. As a consequence, Goldstone commutator: $[Q_{T}, O(x)]$ for any composite operator $O(\phi(x),\Pi(x))$ also vanishes. This means that in all possible vacua, there would be no Nambu-Goldstone gapless mode associated with charge $Q_{T}$, since nonvanishing of vacuum expectation value of Goldstone commutator is a necessary and sufficient condition for their existence \cite{ojima,stro}. In absence of this gapless mode, gauge field $A_{\mu}$ can not attain mass via Anderson-Higgs mechanism.   
\end{proof}  
This conclusion was also obtained in Ref. \cite{shamir}. This theorem gives rise to several interesting consequences:   

\begin{corollary} Consider a theory of massive spinors coupled to a Abelian gauge field, 
	\begin{equation} \nonumber
	\mathscr{L} = \bar{\psi} (i\gamma^{\mu} \partial_{\mu} - m) \psi + \mathscr{L}_{int}(\bar{\psi},\psi) + \mathscr{L}_{g}(\bar{\psi},\psi,A_{\mu}),
	\end{equation}
	where $\mathscr{L}_{int}$ interaction terms involving only spinor fields, whereas $\mathscr{L}_{g}$ stands for gauge-spinor interaction terms (It is assumed that the net action is gauge invariant under appropriate local gauge transformations). Further, assume that above Lagrangian is invariant under symmetry transformation $\psi \rightarrow \psi e^{-i\theta}$, where $\theta$ is a real parameter. Then in such a theory, gauge boson can gain mass via Anderson-Higgs mechanism, only if there is minimal coupling $\bar{\psi} \gamma^{\mu}\psi A_{\mu}$ between fermions and gauge field.
\end{corollary}

\begin{proof} Above theory possesses only one internal symmetry, which is generated by conserved Noether charge $Q$, corresponding to current $\bar{\psi} \gamma^{\mu}\psi$. Hence existence of any other conserved current in this theory has to be of topological origin\footnote{Here we are not considering currents which are conserved due to space-time symmetries.}. Now if gauge field is only coupled to any of such topological currents, then as per Theorem \ref{t3}, gauge field can not get mass via Anderson-Higgs route.            
\end{proof}

It is straightforward to see that a similar result will also hold if one has non-Abelian gauge field coupled to fermions. As is evident this corollary puts severe constraint on the type of theory one can construct involving massive gauge bosons, where gauge boson mass comes via Anderson-Higgs route. In particular, this corollary has an interesting consequence for non-minimally coupled QED in $3+1$ dimensional spacetime:
\begin{equation} \nonumber
\mathscr{L} = \bar{\psi} (i\gamma^{\mu} \partial_{\mu} - m) \psi + g \bar{\psi} \sigma^{\mu \nu} \psi F_{\mu \nu} + \mathscr{L}_{int}(\bar{\psi},\psi),
\end{equation}
where $\mathscr{L}_{int}$ represents other interaction terms involving only spinor fields. Here the gauge field couples to conserved spin current $J^{\mu} = \partial_{\nu} (\bar{\psi} \sigma^{\mu \nu} \psi)$ which, as noted before, is topological. In light of above corollary, this model can not exhibit superconductivity via Anderson-Higgs mechanism. Such non-minimally coupled models are relevant in context of physics of neutrons, since neutrons being electrically neutral, can only interact with photons non-minimally \cite{nair}. The same also holds for systems involving Majorana fermions, which are currently studied intensely in context of condensed matter systems \cite{tsv,qi}. These results may also have relevance for condensed matter systems, where the coupling of quasiparticles to electromagnetic field is only via Zeeman term, for example in models dealing with magnetism and high $T_{c}$ superconductivity \cite{tsv}.

It is not difficult to see that, similar corollary will hold for any interacting theory of scalar fields, namely gauge boson can gain mass via Anderson-Higgs mechanism only if scalar field is minimally coupled to gauge field. A similar corollary will also hold for both nonrelativistic theory of interactings bosons and fermions in general. Infact, nonrelativistic theory of interacting bosons living on a line, which was discussed earlier, has found realisation in ultracold atom experiments \cite{dal,tongs}. Since the constituent atoms in these experiments are electrically neutral, the effective theory of bosons arising out of them does not couple to electromagnetic field minimally \cite{dal,pol}. In this light, above results will have relevance to these systems as well.

Apart from above theories, similar conclusions will also hold for theories involving antisymmetric tensor field $B_{\mu \nu}$ (in $3+1$ dimensions), wherein the coupling between gauge field $A_{\mu}$ and $B_{\mu \nu}$, is of $B \wedge F$ form: $\epsilon^{\mu \nu \rho \lambda} B_{\mu \nu} F_{\rho \lambda}$. As is evident this coupling is non-minimal, and the gauge field is coupled to a topological current. Hence, from Theorem \ref{t3}, it follows that $A_{\mu}$ can not get mass via Anderson-Higgs route. However, it is worth emphasing that $A_{\mu}$ can still be massive, as seen in Ref. \cite{leblanc}, without any spontaneous symmetry breaking whatsoever.               

Above results, in certain sense, indicate the uniqueness of minimal gauge coupling in quantum field theory.

\section{Extended objects and classical theory}\label{nonrel}          

As mentioned earlier, a quantum field theory, unlike quantum mechanics of N-point particles, can realise different vacua, which are unitarily inequivalent to each other \cite{ume2,ume}. In many cases, the system realises a vacuum such that vacuum expectation value of field is non-vanishing and in general a function of spacetime $\langle vac | \Phi(x)| vac \rangle = f(x) \neq 0$. In different context, such space-time dependent vacua are known by different names like solitons, vortices, kinks, hedgehogs, monopoles and so on \cite{ume2}. Generically these are also called extended objects, since function $f(x)$ is non-vanishing in finite spatial domain, or equivalently has a finite spatial extent. 

It is seen in many field theories that the extended objects carry topological charge. To be more precise, for some topological charge $Q_{top.}$ in the theory, there exists a particular vacuum $|f(x)\rangle$ such that, $\Phi(x)| f(x) \rangle = f(x) | f(x) \rangle $ and $Q_{top.} |f(x)\rangle \neq 0$. A well known example of such topological object is a vortex vacuum of real scalar theory in $2+1$ dimensions, for which the topological charge $Q= \int d^{2}x \: \vec{\nabla} \times \vec{\nabla} \Phi = \oint_{c} d\vec{r} \cdot \vec{\nabla} \Phi$ is non-vanishing \cite{ume2}. An important consequence follows from these; in event of instability of a vacuum which has non-vanishing topological charge, the theory can not realise another vacuum, which has a topological charge different than the former, since the topological charge is conserved in the theory. Hence topological charge makes such extended objects extraordinarily stable,  and it gives rise to various physical phenomena \cite{ume2,ume}. In many cases, such topological objects lead to fermion number and charge fractionalisation \cite{jackiw,rao}.

Above discussion convinces one that, the study of topological charges in vacuum sector or equivalently in the classical theory is interesting in its own right, and can lead to valuable information about the theory. As an example, below we study an integrable classical field theory model, and show that ideas similar to quantum case, lead to important results regarding solution space of the theory.

Consider a nonrelavistic classical field theory, living on a line, involving a complex field $\psi(x,t)$  defined by the Lagrangian:
\begin{equation} \nonumber
\mathscr{L} = \frac{i}{2} \left( \psi \partial_{t} \psi^{\ast} - \psi^{\ast} \partial_{t} \psi \right)
- \frac{1}{2} \left( \partial_{x} \psi^{\ast} \partial_{x} \psi + |\psi|^{4} \right). 
\end{equation}
The equation of motion is given by the nonlinear Schr\"odinger equation \cite{das2}:
\begin{equation} \nonumber
i \partial_{t} \psi + \frac{1}{2} \partial_{xx} \psi - |\psi|^{2} \psi = 0, 
\end{equation}          
which appears in diverse areas like optics, plasma physics, hydrodynamics and so on \cite{kuznetsov,kivshar}. From the monumental work of Zakharov and Shabat \cite{zakharov}, it became known that nonlinear Schr\"odinger equation is an integrable system, in the sense that, there exists infinitely many independently conserved quantities. Integrability of this system has many interesting connections with two dimensional conformal field theory and gravity \cite{das2,polyakov,lina}. The canonical equal-time Poisson bracket for this system is given by: 
$\left\lbrace \psi(x,t), \psi^{\ast}(y,t) \right\rbrace = i \delta (x -y)$. From the discussion in section (\ref{tc}) regarding nonrelativistic theory, it is clear that this theory admits, amongst others, following three topological charges:
\begin{equation} \nonumber
Q_{1} = \mathrm{Re} \; \psi \biggr\rvert_{-\infty}^{\infty},\: Q_{2} = \mathrm{Im} \; \psi \biggr\rvert_{-\infty}^{\infty}, \: \text{and} \: Q_{3} = \psi^{\ast}\psi \biggr\rvert_{-\infty}^{\infty}.
\end{equation}     
It is well known that above equation of motion admits, what is called kink soliton or dark soliton \cite{frantz} as a solution:
\begin{equation} \nonumber
\Psi(x,t) = \sqrt{n_{0}} \left[ B \mathrm{tanh} (\xi) + i A \right] e^{-i n_{0} t},
\end{equation}           
where $\xi = \sqrt{n_{0}} B (x - A \sqrt{n_{0}} t)$ and $A^{2} + B^{2} = 1$ ($n_{0}$, A, B are real parameters). It is straightforward to see that, this kink soliton carries $Q_{1}$ topological charge, infact, $Q_{1}[\Psi] = 2$, and $Q_{2,3}[\Psi] = 0$. Interestingly, however if one considers above solution with an additional overall phase of $-\frac{\pi}{2}$, then one sees that, $Q_{2}[\Psi] = 2$, and $Q_{1,3}[\Psi] = 0$. In any case, this clearly shows that \emph{kink solitons are topological objects}, and in general carry charges $Q_{1,2}$ as defined above. It is worth mentioning that these solitons are known in the literature since 1970's and have been extensively studied both theoretically and experimentally \cite{kivshar,frantz,carr}, however their topological nature was hitherto not understood. From above discussion, one is able to see clearly the topological nature of these solitons, and the reason behind their stability. 


A variety of variants of nonlinear Schr\"odinger equation also exist in the literature, which are studied in context of various physical phenomena, which possess such extended solutions as above mentioned soliton (for example see \cite{soloman,vyas}). Note that since charges $Q_{1,2,3}$ defined above are topological, their conservation holds for all variants of nonlinear Schr\"odinger equation (and of course for usual linear Schr\"odinger equation as well). 
In particular, we would like to point out that many solutions found in Ref. \cite{cnk}, while studying a certain variant of nonlinear Schr\"odinger equation, carry topological charge $Q_{3}$ unlike kink solitons encountered above.

\section{Conclusion}          

In this paper, topological currents occurring in field theory (both classical as well as quantum) are studied. It is proven in generality that, topological charges commute with dynamical fields in any given theory, with some reasonable assumptions. It is also conversely proven that, if there exists a charge in the theory which commutes with all the dynamical fields of the theory, then the corresponding current has to be topological. Interestingly, this implies that unlike many Noether charges, topological charge can not have any explicit time dependence. As a consequence of above results, it is seen that, each conserved current can be classified exactly as either being a Noether current, or a topological current. It is clear that these two currents are of fundamentally different character, in the sense that, Noether charges generate non-trivial symmetry transformations preserving the equations of motion, unlike the ones generated by topological charges which are trivial. This gives rise to many interesting consequences. For example, any modification of Noether charges by addition of topological charges, does not alter the algebra of Noether charges, and hence symmetry properties of the theory. This implies that Noether charges, which differ by a topological charge are equivalent, and the space of all Noether charges is a quotient space.  

It is also proven that, when a gauge field in a given theory, only couples to topological currents, then the gauge field can never become massive via Anderson-Higgs mechanism. This has implications to situations where the matter-gauge coupling is non-minimal, like in many models in high energy and condensed matter physics. However, it must be pointed out that, above result, in case where it is applicable, do not prohibit existence of gauge boson mass or superconductivity. It only asserts that gauge boson mass can not arise via Anderson-Higgs mechanism. However, it is known that gauge boson mass can arise without Anderson-Higgs mechanism, as for example in Schwinger model, and above result does not negate such a possibility \cite{sch,das}. A connection of extended objects with topological currents is shown, and a nonrelativistic self interacting model is studied. It is shown that, the well known kink solitons of this model, are actually topological. Its implications to other variants of this model are also discussed.

\section*{Acknowledgements}
VMV thanks Dr. Nitin Chandra for several useful discussions. VS thanks Prof. Rita John and the Department of Theoretical Physics, University of Madras for their hospitality. The authors acknowledge a communication from Prof. Roman Jackiw regarding this work and Noether theorems.


\begin{thebibliography}{68}%
	\makeatletter
	\providecommand \@ifxundefined [1]{%
		\@ifx{#1\undefined}
	}%
	\providecommand \@ifnum [1]{%
		\ifnum #1\expandafter \@firstoftwo
		\else \expandafter \@secondoftwo
		\fi
	}%
	\providecommand \@ifx [1]{%
		\ifx #1\expandafter \@firstoftwo
		\else \expandafter \@secondoftwo
		\fi
	}%
	\providecommand \natexlab [1]{#1}%
	\providecommand \enquote  [1]{``#1''}%
	\providecommand \bibnamefont  [1]{#1}%
	\providecommand \bibfnamefont [1]{#1}%
	\providecommand \citenamefont [1]{#1}%
	\providecommand \href@noop [0]{\@secondoftwo}%
	\providecommand \href [0]{\begingroup \@sanitize@url \@href}%
	\providecommand \@href[1]{\@@startlink{#1}\@@href}%
	\providecommand \@@href[1]{\endgroup#1\@@endlink}%
	\providecommand \@sanitize@url [0]{\catcode `\\12\catcode `\$12\catcode
		`\&12\catcode `\#12\catcode `\^12\catcode `\_12\catcode `\%12\relax}%
	\providecommand \@@startlink[1]{}%
	\providecommand \@@endlink[0]{}%
	\providecommand \url  [0]{\begingroup\@sanitize@url \@url }%
	\providecommand \@url [1]{\endgroup\@href {#1}{\urlprefix }}%
	\providecommand \urlprefix  [0]{URL }%
	\providecommand \Eprint [0]{\href }%
	\providecommand \doibase [0]{http://dx.doi.org/}%
	\providecommand \selectlanguage [0]{\@gobble}%
	\providecommand \bibinfo  [0]{\@secondoftwo}%
	\providecommand \bibfield  [0]{\@secondoftwo}%
	\providecommand \translation [1]{[#1]}%
	\providecommand \BibitemOpen [0]{}%
	\providecommand \bibitemStop [0]{}%
	\providecommand \bibitemNoStop [0]{.\EOS\space}%
	\providecommand \EOS [0]{\spacefactor3000\relax}%
	\providecommand \BibitemShut  [1]{\csname bibitem#1\endcsname}%
	\let\auto@bib@innerbib\@empty
	\bibitem [{\citenamefont {Nair}(2005)}]{nair}%
	\BibitemOpen
	\bibfield  {author} {\bibinfo {author} {\bibfnamefont {V.~P.}\ \bibnamefont
			{Nair}},\ }\href@noop {} {\emph {\bibinfo {title} {Quantum field theory: A
				modern perspective}}}\ (\bibinfo  {publisher} {Springer},\ \bibinfo {year}
	{2005})\BibitemShut {NoStop}%
	\bibitem [{\citenamefont {Jackiw}\ and\ \citenamefont
		{Schrieffer}(1981)}]{jackiw}%
	\BibitemOpen
	\bibfield  {author} {\bibinfo {author} {\bibfnamefont {R.}~\bibnamefont
			{Jackiw}}\ and\ \bibinfo {author} {\bibfnamefont {J.~R.}\ \bibnamefont
			{Schrieffer}},\ }\href@noop {} {\bibfield  {journal} {\bibinfo  {journal}
			{Nuc. Phys. B}\ }\textbf {\bibinfo {volume} {190}},\ \bibinfo {pages} {253}
		(\bibinfo {year} {1981})}\BibitemShut {NoStop}%
	\bibitem [{\citenamefont {Schwarz}\ \emph {et~al.}(1993)\citenamefont
		{Schwarz}, \citenamefont {Yankowsky},\ and\ \citenamefont {Levy}}]{schw}%
	\BibitemOpen
	\bibfield  {author} {\bibinfo {author} {\bibfnamefont {A.~S.}\ \bibnamefont
			{Schwarz}}, \bibinfo {author} {\bibfnamefont {E.}~\bibnamefont {Yankowsky}},
		\ and\ \bibinfo {author} {\bibfnamefont {S.}~\bibnamefont {Levy}},\
	}\href@noop {} {\emph {\bibinfo {title} {Quantum field theory and
				topology}}}\ (\bibinfo  {publisher} {Springer},\ \bibinfo {year}
	{1993})\BibitemShut {NoStop}%
	\bibitem [{\citenamefont {Umezawa}(1995)}]{ume2}%
	\BibitemOpen
	\bibfield  {author} {\bibinfo {author} {\bibfnamefont {H.}~\bibnamefont
			{Umezawa}},\ }\href@noop {} {\emph {\bibinfo {title} {Advanced field theory:
				micro, macro, and thermal physics}}}\ (\bibinfo  {publisher} {American
		Institute of Physics},\ \bibinfo {year} {1995})\BibitemShut {NoStop}%
	\bibitem [{\citenamefont {Tsvelik}(2007)}]{tsv}%
	\BibitemOpen
	\bibfield  {author} {\bibinfo {author} {\bibfnamefont {A.~M.}\ \bibnamefont
			{Tsvelik}},\ }\href@noop {} {\emph {\bibinfo {title} {Quantum field theory in
				condensed matter physics}}}\ (\bibinfo  {publisher} {Cambridge University
		Press},\ \bibinfo {year} {2007})\BibitemShut {NoStop}%
	\bibitem [{Note1()}]{Note1}%
	\BibitemOpen
	\bibinfo {note} {Throughout this paper we shall denote n-dimensional spatial
		volume element by $dV_{x}=dx_{1}dx_{2} \protect \cdots dx_{n}$. Sometimes for
		brevity subscript may be dropped in $dV_{x}$. Our metric convention is $\eta
		_{\mu \nu } = diag(1,-1,-1, \protect \cdots , -1)$. The notations and
		conventions of that of Itzykson and Zuber \cite {iz} are followed, unless
		mentioned otherwise.}\BibitemShut {Stop}%
	\bibitem [{\citenamefont {Kosmann-Schwarzbach}(2011)}]{kosmann}%
	\BibitemOpen
	\bibfield  {author} {\bibinfo {author} {\bibfnamefont {Y.}~\bibnamefont
			{Kosmann-Schwarzbach}},\ }\href@noop {} {\emph {\bibinfo {title} {The Noether
				theorems}}}\ (\bibinfo  {publisher} {Springer},\ \bibinfo {year}
	{2011})\BibitemShut {NoStop}%
	\bibitem [{\citenamefont {Brown}\ and\ \citenamefont {Holland}(2004)}]{brown}%
	\BibitemOpen
	\bibfield  {author} {\bibinfo {author} {\bibfnamefont {H.~R.}\ \bibnamefont
			{Brown}}\ and\ \bibinfo {author} {\bibfnamefont {P.}~\bibnamefont
			{Holland}},\ }\href@noop {} {\bibfield  {journal} {\bibinfo  {journal}
			{Molecular Physics}\ }\textbf {\bibinfo {volume} {102}},\ \bibinfo {pages}
		{1133} (\bibinfo {year} {2004})}\BibitemShut {NoStop}%
	\bibitem [{\citenamefont {Olver}(2000)}]{olver}%
	\BibitemOpen
	\bibfield  {author} {\bibinfo {author} {\bibfnamefont {P.~J.}\ \bibnamefont
			{Olver}},\ }\href@noop {} {\emph {\bibinfo {title} {Applications of Lie
				groups to differential equations}}},\ Vol.\ \bibinfo {volume} {107}\
	(\bibinfo  {publisher} {Springer},\ \bibinfo {year} {2000})\BibitemShut
	{NoStop}%
	\bibitem [{Note2()}]{Note2}%
	\BibitemOpen
	\bibinfo {note} {Here $\protect \frac {d}{dx^{\mu }} = d_{\mu }$ stands for
		total derivative with respect to $x^{\mu }$, which takes into account both
		implicit and explicit $x^{\mu }$ dependence, unlike $\partial _{\mu }$ which
		is only concerned with explicit $x^{\mu }$ dependence \cite {brown,rosen}.
		For example: $d_{\mu } f(x) = \partial _{\mu } f(x)$ and $d_{\mu } f(x, g(x))
		= \partial _{\mu } f(x,g(x)) + \partial _{\mu }g(x) \protect \frac {\partial
		}{\partial g(x)} f(x,g(x))$.}\BibitemShut {Stop}%
	\bibitem [{Note3()}]{Note3}%
	\BibitemOpen
	\bibinfo {note} {The statement that, as $|\protect \mathaccentV {vec}17E{x}|
		\rightarrow \infty $, current $\protect \mathaccentV {vec}17E{j}(x)$ decays
		sufficiently fast, in strict sense, is a statement about the kind of boundary
		conditions the system is assumed to obey.}\BibitemShut {Stop}%
	\bibitem [{\citenamefont {Schwinger}(1951)}]{schwinger}%
	\BibitemOpen
	\bibfield  {author} {\bibinfo {author} {\bibfnamefont {J.}~\bibnamefont
			{Schwinger}},\ }\href@noop {} {\bibfield  {journal} {\bibinfo  {journal}
			{Phys. Rev.}\ }\textbf {\bibinfo {volume} {82}},\ \bibinfo {pages} {914}
		(\bibinfo {year} {1951})}\BibitemShut {NoStop}%
	\bibitem [{\citenamefont {Rosen}(1972)}]{rosen}%
	\BibitemOpen
	\bibfield  {author} {\bibinfo {author} {\bibfnamefont {J.}~\bibnamefont
			{Rosen}},\ }\href@noop {} {\bibfield  {journal} {\bibinfo  {journal} {Ann.
				Phys.}\ }\textbf {\bibinfo {volume} {69}},\ \bibinfo {pages} {349} (\bibinfo
		{year} {1972})}\BibitemShut {NoStop}%
	\bibitem [{Note4()}]{Note4}%
	\BibitemOpen
	\bibinfo {note} {Olver \cite {olver} refers to such conservation laws as null
		divergences, and proves a theorem which shows their connection with
		Poincar\'{e} lemma. {For discussions on the role of topology in this context,
			readers are referred to Ref. \cite {nair} and Ref. \cite
			{patani}.}}\BibitemShut {Stop}%
	\bibitem [{Note5()}]{Note5}%
	\BibitemOpen
	\bibinfo {note} {Topological currents can be studied using the coordinate
		free language of differential forms, as is done in detail in Ref. \cite
		{patani}. Such a study brings to fore the global topological aspects of these
		currents.}\BibitemShut {Stop}%
	\bibitem [{Note6()}]{Note6}%
	\BibitemOpen
	\bibinfo {note} {Topological currents are intimately connected with local
		gauge invariances, which was shown by E. Noether, (not so wellknown \protect
		\emph {Noether Second Theorem}) while working within the Lagrangian framework
		\cite {brading}. In this paper, the discussion is confined within Hamiltonian
		framework, which does not admit local invariances. Nevertheless it would be
		interesting to explore the implications of the Noether second theorem in
		canonical framework. Recently its implications in Weyl invariant theories
		have been explored in Ref. \cite {jack} and Ref. \cite {shukla}}\BibitemShut
	{NoStop}%
	\bibitem [{\citenamefont {Takahashi}\ and\ \citenamefont
		{Umezawa}(1964)}]{taka}%
	\BibitemOpen
	\bibfield  {author} {\bibinfo {author} {\bibfnamefont {Y.}~\bibnamefont
			{Takahashi}}\ and\ \bibinfo {author} {\bibfnamefont {H.}~\bibnamefont
			{Umezawa}},\ }\href@noop {} {\bibfield  {journal} {\bibinfo  {journal} {Nuc.
				Phys.}\ }\textbf {\bibinfo {volume} {51}},\ \bibinfo {pages} {193} (\bibinfo
		{year} {1964})}\BibitemShut {NoStop}%
	\bibitem [{Note7()}]{Note7}%
	\BibitemOpen
	\bibinfo {note} {One can consider $F$ to be any general local function of
		dynamical fields, here polynomial function is considered for the sake of
		simplicity.}\BibitemShut {Stop}%
	\bibitem [{\citenamefont {Ezawa}(1978)}]{ezawa}%
	\BibitemOpen
	\bibfield  {author} {\bibinfo {author} {\bibfnamefont {Z.}~\bibnamefont
			{Ezawa}},\ }\href@noop {} {\bibfield  {journal} {\bibinfo  {journal} {Phys.
				Rev. D}\ }\textbf {\bibinfo {volume} {18}},\ \bibinfo {pages} {2091}
		(\bibinfo {year} {1978})}\BibitemShut {NoStop}%
	\bibitem [{\citenamefont {Wick}\ \emph {et~al.}(1952)\citenamefont {Wick},
		\citenamefont {Wightman},\ and\ \citenamefont {Wigner}}]{wick}%
	\BibitemOpen
	\bibfield  {author} {\bibinfo {author} {\bibfnamefont {G.}~\bibnamefont
			{Wick}}, \bibinfo {author} {\bibfnamefont {A.}~\bibnamefont {Wightman}}, \
		and\ \bibinfo {author} {\bibfnamefont {E.}~\bibnamefont {Wigner}},\
	}\href@noop {} {\bibfield  {journal} {\bibinfo  {journal} {Physical Review}\
		}\textbf {\bibinfo {volume} {88}},\ \bibinfo {pages} {101} (\bibinfo {year}
		{1952})}\BibitemShut {NoStop}%
	\bibitem [{\citenamefont {Wightman}(1995)}]{wightman}%
	\BibitemOpen
	\bibfield  {author} {\bibinfo {author} {\bibfnamefont {A.}~\bibnamefont
			{Wightman}},\ }\href@noop {} {\bibfield  {journal} {\bibinfo  {journal} {Il
				Nuovo Cimento B}\ }\textbf {\bibinfo {volume} {110}},\ \bibinfo {pages} {751}
		(\bibinfo {year} {1995})}\BibitemShut {NoStop}%
	\bibitem [{\citenamefont {Shamir}\ and\ \citenamefont {Park}(1991)}]{shamir}%
	\BibitemOpen
	\bibfield  {author} {\bibinfo {author} {\bibfnamefont {Y.}~\bibnamefont
			{Shamir}}\ and\ \bibinfo {author} {\bibfnamefont {S.~H.}\ \bibnamefont
			{Park}},\ }\href@noop {} {\bibfield  {journal} {\bibinfo  {journal} {Phys.
				Lett. B}\ }\textbf {\bibinfo {volume} {258}},\ \bibinfo {pages} {179}
		(\bibinfo {year} {1991})}\BibitemShut {NoStop}%
	\bibitem [{\citenamefont {Henneaux}\ and\ \citenamefont
		{Teitelboim}(1992)}]{hen}%
	\BibitemOpen
	\bibfield  {author} {\bibinfo {author} {\bibfnamefont {M.}~\bibnamefont
			{Henneaux}}\ and\ \bibinfo {author} {\bibfnamefont {C.}~\bibnamefont
			{Teitelboim}},\ }\href@noop {} {\emph {\bibinfo {title} {Quantization of
				gauge systems}}}\ (\bibinfo  {publisher} {Princeton university press},\
	\bibinfo {year} {1992})\BibitemShut {NoStop}%
	\bibitem [{\citenamefont {Nakanishi}(1980)}]{naka}%
	\BibitemOpen
	\bibfield  {author} {\bibinfo {author} {\bibfnamefont {N.}~\bibnamefont
			{Nakanishi}},\ }\href@noop {} {\bibfield  {journal} {\bibinfo  {journal}
			{Zeitschrift f{\"u}r Physik C Particles and Fields}\ }\textbf {\bibinfo
			{volume} {4}},\ \bibinfo {pages} {17} (\bibinfo {year} {1980})}\BibitemShut
	{NoStop}%
	\bibitem [{\citenamefont {Nakanishi}\ and\ \citenamefont
		{Ojima}(1990)}]{ojima}%
	\BibitemOpen
	\bibfield  {author} {\bibinfo {author} {\bibfnamefont {N.}~\bibnamefont
			{Nakanishi}}\ and\ \bibinfo {author} {\bibfnamefont {I.}~\bibnamefont
			{Ojima}},\ }\href@noop {} {\emph {\bibinfo {title} {Covariant operator
				formalism of gauge theories and quantum gravity}}}\ (\bibinfo  {publisher}
	{World Scientific},\ \bibinfo {year} {1990})\BibitemShut {NoStop}%
	\bibitem [{Note8()}]{Note8}%
	\BibitemOpen
	\bibinfo {note} {One notes an interesting parallel of this with that of
		differentiation under integral sign: $\protect \frac {d}{dt} \DOTSI \intop
		\ilimits@ _{a(t)}^{b(t)} dx \protect \tmspace +\thickmuskip {.2777em} f(x,t)
		= \DOTSI \intop \ilimits@ _{a(t)}^{b(t)} dx \protect \tmspace +\thickmuskip
		{.2777em} \protect \frac {\partial }{\partial t} f(x,t) + f(b(t),t) b'(t) -
		f(a(t),t) a'(t)$, by identifying differentiation with derivation: $\protect
		\frac {\partial }{\partial t} (\protect \cdots ) \leftrightarrow [Q,
		(\protect \cdots )]$.}\BibitemShut {Stop}%
	\bibitem [{\citenamefont {Mandelstam}(1975)}]{mandel}%
	\BibitemOpen
	\bibfield  {author} {\bibinfo {author} {\bibfnamefont {S.}~\bibnamefont
			{Mandelstam}},\ }\href@noop {} {\bibfield  {journal} {\bibinfo  {journal}
			{Physical Review D}\ }\textbf {\bibinfo {volume} {11}},\ \bibinfo {pages}
		{3026} (\bibinfo {year} {1975})}\BibitemShut {NoStop}%
	\bibitem [{\citenamefont {Barcelos-Neto}\ \emph {et~al.}(1986)\citenamefont
		{Barcelos-Neto}, \citenamefont {Das},\ and\ \citenamefont {Maharana}}]{bar}%
	\BibitemOpen
	\bibfield  {author} {\bibinfo {author} {\bibfnamefont {J.}~\bibnamefont
			{Barcelos-Neto}}, \bibinfo {author} {\bibfnamefont {A.}~\bibnamefont {Das}},
		\ and\ \bibinfo {author} {\bibfnamefont {J.}~\bibnamefont {Maharana}},\
	}\href@noop {} {\bibfield  {journal} {\bibinfo  {journal} {Zeitschrift
				f{\"u}r Physik C Particles and Fields}\ }\textbf {\bibinfo {volume} {30}},\
		\bibinfo {pages} {401} (\bibinfo {year} {1986})}\BibitemShut {NoStop}%
	\bibitem [{\citenamefont {Sakurai}(1967)}]{sakurai}%
	\BibitemOpen
	\bibfield  {author} {\bibinfo {author} {\bibfnamefont {J.~J.}\ \bibnamefont
			{Sakurai}},\ }\href@noop {} {\emph {\bibinfo {title} {Advanced quantum
				mechanics}}}\ (\bibinfo  {publisher} {Addison-Wesley},\ \bibinfo {year}
	{1967})\BibitemShut {NoStop}%
	\bibitem [{\citenamefont {Hehl}\ \emph {et~al.}(1991)\citenamefont {Hehl},
		\citenamefont {Lemke},\ and\ \citenamefont {Mielke}}]{hehl}%
	\BibitemOpen
	\bibfield  {author} {\bibinfo {author} {\bibfnamefont {F.~W.}\ \bibnamefont
			{Hehl}}, \bibinfo {author} {\bibfnamefont {J.}~\bibnamefont {Lemke}}, \ and\
		\bibinfo {author} {\bibfnamefont {E.~W.}\ \bibnamefont {Mielke}},\ }in\
	\href@noop {} {\emph {\bibinfo {booktitle} {Geometry and theoretical
				physics}}}\ (\bibinfo  {publisher} {Springer},\ \bibinfo {year} {1991})\ pp.\
	\bibinfo {pages} {56--140}\BibitemShut {NoStop}%
	\bibitem [{\citenamefont {Greiner}\ and\ \citenamefont
		{Reinhardt}(1996)}]{greiner}%
	\BibitemOpen
	\bibfield  {author} {\bibinfo {author} {\bibfnamefont {W.}~\bibnamefont
			{Greiner}}\ and\ \bibinfo {author} {\bibfnamefont {J.}~\bibnamefont
			{Reinhardt}},\ }\href@noop {} {\emph {\bibinfo {title} {Field
				quantization}}}\ (\bibinfo  {publisher} {Springer},\ \bibinfo {year}
	{1996})\BibitemShut {NoStop}%
	\bibitem [{Note9()}]{Note9}%
	\BibitemOpen
	\bibinfo {note} {Redefinition of energy-momentum tensor and angular momemtum
		tensor is also discussed in detail in Ref. \cite {hehl}}\BibitemShut
	{NoStop}%
	\bibitem [{\citenamefont {Deser}\ \emph {et~al.}(1982)\citenamefont {Deser},
		\citenamefont {Jackiw},\ and\ \citenamefont {Templeton}}]{deser}%
	\BibitemOpen
	\bibfield  {author} {\bibinfo {author} {\bibfnamefont {S.}~\bibnamefont
			{Deser}}, \bibinfo {author} {\bibfnamefont {R.}~\bibnamefont {Jackiw}}, \
		and\ \bibinfo {author} {\bibfnamefont {S.}~\bibnamefont {Templeton}},\
	}\href@noop {} {\bibfield  {journal} {\bibinfo  {journal} {Ann. Phys.}\
		}\textbf {\bibinfo {volume} {140}},\ \bibinfo {pages} {372} (\bibinfo {year}
		{1982})}\BibitemShut {NoStop}%
	\bibitem [{\citenamefont {Halmos}(1947)}]{halmos}%
	\BibitemOpen
	\bibfield  {author} {\bibinfo {author} {\bibfnamefont {P.~R.}\ \bibnamefont
			{Halmos}},\ }\href@noop {} {\emph {\bibinfo {title} {Finite dimensional
				vector spaces}}},\ \bibinfo {number} {7}\ (\bibinfo  {publisher} {Princeton
		University Press},\ \bibinfo {year} {1947})\BibitemShut {NoStop}%
	\bibitem [{Note10()}]{Note10}%
	\BibitemOpen
	\bibinfo {note} {Here $\sigma $ denotes the space-like surface over which the
		integration is to be done \cite {schwinger}, generally taken to be a constant
		time surface. Here $ds_{\mu }$ stands for oriented area element on the
		surface $\sigma $.}\BibitemShut {Stop}%
	\bibitem [{\citenamefont {Jackiw}(1995)}]{jackiwd}%
	\BibitemOpen
	\bibfield  {author} {\bibinfo {author} {\bibfnamefont {R.~W.}\ \bibnamefont
			{Jackiw}},\ }\href@noop {} {\emph {\bibinfo {title} {Diverse topics in
				theoretical and mathematical physics}}}\ (\bibinfo  {publisher} {World
		Scientific},\ \bibinfo {year} {1995})\BibitemShut {NoStop}%
	\bibitem [{Note11()}]{Note11}%
	\BibitemOpen
	\bibinfo {note} {Mathematically this is one point compactification \cite
		{munkres} of Euclidean space-time $\protect \mathbb {R}^{n}$ to n-sphere
		$\protect \mathbb {S}^{n}$.}\BibitemShut {Stop}%
	\bibitem [{Note12()}]{Note12}%
	\BibitemOpen
	\bibinfo {note} {This closed space-like surface is also a sphere, albeit of
		one lesser dimension than the space-time sphere.}\BibitemShut {Stop}%
	\bibitem [{\citenamefont {Umezawa}\ \emph {et~al.}(1982)\citenamefont
		{Umezawa}, \citenamefont {Matsumoto},\ and\ \citenamefont {Tachiki}}]{ume}%
	\BibitemOpen
	\bibfield  {author} {\bibinfo {author} {\bibfnamefont {H.}~\bibnamefont
			{Umezawa}}, \bibinfo {author} {\bibfnamefont {H.}~\bibnamefont {Matsumoto}},
		\ and\ \bibinfo {author} {\bibfnamefont {M.}~\bibnamefont {Tachiki}},\
	}\href@noop {} {\emph {\bibinfo {title} {Thermo field dynamics and condensed
				states.}}}\ (\bibinfo  {publisher} {North-Holland},\ \bibinfo {year}
	{1982})\BibitemShut {NoStop}%
	\bibitem [{\citenamefont {Goldstone}\ \emph {et~al.}(1962)\citenamefont
		{Goldstone}, \citenamefont {Salam},\ and\ \citenamefont
		{Weinberg}}]{goldstone}%
	\BibitemOpen
	\bibfield  {author} {\bibinfo {author} {\bibfnamefont {J.}~\bibnamefont
			{Goldstone}}, \bibinfo {author} {\bibfnamefont {A.}~\bibnamefont {Salam}}, \
		and\ \bibinfo {author} {\bibfnamefont {S.}~\bibnamefont {Weinberg}},\
	}\href@noop {} {\bibfield  {journal} {\bibinfo  {journal} {Phys. Rev.}\
		}\textbf {\bibinfo {volume} {127}},\ \bibinfo {pages} {965} (\bibinfo {year}
		{1962})}\BibitemShut {NoStop}%
	\bibitem [{\citenamefont {Nambu}(1960)}]{nambu}%
	\BibitemOpen
	\bibfield  {author} {\bibinfo {author} {\bibfnamefont {Y.}~\bibnamefont
			{Nambu}},\ }\href@noop {} {\bibfield  {journal} {\bibinfo  {journal} {Phys.
				Rev.}\ }\textbf {\bibinfo {volume} {117}},\ \bibinfo {pages} {648} (\bibinfo
		{year} {1960})}\BibitemShut {NoStop}%
	\bibitem [{\citenamefont {Strocchi}(1977)}]{stro}%
	\BibitemOpen
	\bibfield  {author} {\bibinfo {author} {\bibfnamefont {F.}~\bibnamefont
			{Strocchi}},\ }\href@noop {} {\bibfield  {journal} {\bibinfo  {journal}
			{Comm. Math. Phys.}\ }\textbf {\bibinfo {volume} {56}},\ \bibinfo {pages}
		{57} (\bibinfo {year} {1977})}\BibitemShut {NoStop}%
	\bibitem [{Note13()}]{Note13}%
	\BibitemOpen
	\bibinfo {note} {Here we are not considering currents which are conserved due
		to space-time symmetries.}\BibitemShut {Stop}%
	\bibitem [{\citenamefont {Qi}\ \emph {et~al.}(2009)\citenamefont {Qi},
		\citenamefont {Hughes}, \citenamefont {Raghu},\ and\ \citenamefont
		{Zhang}}]{qi}%
	\BibitemOpen
	\bibfield  {author} {\bibinfo {author} {\bibfnamefont {X.-L.}\ \bibnamefont
			{Qi}}, \bibinfo {author} {\bibfnamefont {T.~L.}\ \bibnamefont {Hughes}},
		\bibinfo {author} {\bibfnamefont {S.}~\bibnamefont {Raghu}}, \ and\ \bibinfo
		{author} {\bibfnamefont {S.-C.}\ \bibnamefont {Zhang}},\ }\href@noop {}
	{\bibfield  {journal} {\bibinfo  {journal} {Phys. Rev. Lett.}\ }\textbf
		{\bibinfo {volume} {102}},\ \bibinfo {pages} {187001} (\bibinfo {year}
		{2009})}\BibitemShut {NoStop}%
	\bibitem [{\citenamefont {Dalfovo}\ \emph {et~al.}(1999)\citenamefont
		{Dalfovo}, \citenamefont {Giorgini}, \citenamefont {Pitaevskii},\ and\
		\citenamefont {Stringari}}]{dal}%
	\BibitemOpen
	\bibfield  {author} {\bibinfo {author} {\bibfnamefont {F.}~\bibnamefont
			{Dalfovo}}, \bibinfo {author} {\bibfnamefont {S.}~\bibnamefont {Giorgini}},
		\bibinfo {author} {\bibfnamefont {L.~P.}\ \bibnamefont {Pitaevskii}}, \ and\
		\bibinfo {author} {\bibfnamefont {S.}~\bibnamefont {Stringari}},\ }\href@noop
	{} {\bibfield  {journal} {\bibinfo  {journal} {Rev. Mod. Phys.}\ }\textbf
		{\bibinfo {volume} {71}},\ \bibinfo {pages} {463} (\bibinfo {year}
		{1999})}\BibitemShut {NoStop}%
	\bibitem [{\citenamefont {Kinoshita}\ \emph {et~al.}(2004)\citenamefont
		{Kinoshita}, \citenamefont {Wenger},\ and\ \citenamefont {Weiss}}]{tongs}%
	\BibitemOpen
	\bibfield  {author} {\bibinfo {author} {\bibfnamefont {T.}~\bibnamefont
			{Kinoshita}}, \bibinfo {author} {\bibfnamefont {T.}~\bibnamefont {Wenger}}, \
		and\ \bibinfo {author} {\bibfnamefont {D.~S.}\ \bibnamefont {Weiss}},\
	}\href@noop {} {\bibfield  {journal} {\bibinfo  {journal} {Science}\ }\textbf
		{\bibinfo {volume} {305}},\ \bibinfo {pages} {1125} (\bibinfo {year}
		{2004})}\BibitemShut {NoStop}%
	\bibitem [{\citenamefont {Politzer}(1991)}]{pol}%
	\BibitemOpen
	\bibfield  {author} {\bibinfo {author} {\bibfnamefont {H.~D.}\ \bibnamefont
			{Politzer}},\ }\href@noop {} {\bibfield  {journal} {\bibinfo  {journal}
			{Phys. Rev. A}\ }\textbf {\bibinfo {volume} {43}},\ \bibinfo {pages} {6444}
		(\bibinfo {year} {1991})}\BibitemShut {NoStop}%
	\bibitem [{\citenamefont {Leblanc}\ \emph {et~al.}(1994)\citenamefont
		{Leblanc}, \citenamefont {MacKenzie}, \citenamefont {Panigrahi},\ and\
		\citenamefont {Ray}}]{leblanc}%
	\BibitemOpen
	\bibfield  {author} {\bibinfo {author} {\bibfnamefont {M.}~\bibnamefont
			{Leblanc}}, \bibinfo {author} {\bibfnamefont {R.}~\bibnamefont {MacKenzie}},
		\bibinfo {author} {\bibfnamefont {P.}~\bibnamefont {Panigrahi}}, \ and\
		\bibinfo {author} {\bibfnamefont {R.}~\bibnamefont {Ray}},\ }\href@noop {}
	{\bibfield  {journal} {\bibinfo  {journal} {Int. J. Mod. Phys. A}\ }\textbf
		{\bibinfo {volume} {9}},\ \bibinfo {pages} {4717} (\bibinfo {year}
		{1994})}\BibitemShut {NoStop}%
	\bibitem [{\citenamefont {Rao}\ \emph {et~al.}(2008)\citenamefont {Rao},
		\citenamefont {Sahu},\ and\ \citenamefont {Panigrahi}}]{rao}%
	\BibitemOpen
	\bibfield  {author} {\bibinfo {author} {\bibfnamefont {K.}~\bibnamefont
			{Rao}}, \bibinfo {author} {\bibfnamefont {N.}~\bibnamefont {Sahu}}, \ and\
		\bibinfo {author} {\bibfnamefont {P.~K.}\ \bibnamefont {Panigrahi}},\
	}\href@noop {} {\bibfield  {journal} {\bibinfo  {journal} {Resonance}\
		}\textbf {\bibinfo {volume} {13}},\ \bibinfo {pages} {738} (\bibinfo {year}
		{2008})}\BibitemShut {NoStop}%
	\bibitem [{\citenamefont {Das}(1989)}]{das2}%
	\BibitemOpen
	\bibfield  {author} {\bibinfo {author} {\bibfnamefont {A.}~\bibnamefont
			{Das}},\ }\href@noop {} {\emph {\bibinfo {title} {Integrable models}}},\
	Vol.~\bibinfo {volume} {30}\ (\bibinfo  {publisher} {World Scientific},\
	\bibinfo {year} {1989})\BibitemShut {NoStop}%
	\bibitem [{\citenamefont {Kuznetsov}\ \emph {et~al.}(1986)\citenamefont
		{Kuznetsov}, \citenamefont {Rubenchik},\ and\ \citenamefont
		{Zakharov}}]{kuznetsov}%
	\BibitemOpen
	\bibfield  {author} {\bibinfo {author} {\bibfnamefont {E.}~\bibnamefont
			{Kuznetsov}}, \bibinfo {author} {\bibfnamefont {A.}~\bibnamefont
			{Rubenchik}}, \ and\ \bibinfo {author} {\bibfnamefont {V.}~\bibnamefont
			{Zakharov}},\ }\href@noop {} {\bibfield  {journal} {\bibinfo  {journal}
			{Phys. Rep.}\ }\textbf {\bibinfo {volume} {142}},\ \bibinfo {pages} {103}
		(\bibinfo {year} {1986})}\BibitemShut {NoStop}%
	\bibitem [{\citenamefont {Kivshar}\ and\ \citenamefont
		{Luther-Davies}(1998)}]{kivshar}%
	\BibitemOpen
	\bibfield  {author} {\bibinfo {author} {\bibfnamefont {Y.~S.}\ \bibnamefont
			{Kivshar}}\ and\ \bibinfo {author} {\bibfnamefont {B.}~\bibnamefont
			{Luther-Davies}},\ }\href@noop {} {\bibfield  {journal} {\bibinfo  {journal}
			{Phys. Rep.}\ }\textbf {\bibinfo {volume} {298}},\ \bibinfo {pages} {81}
		(\bibinfo {year} {1998})}\BibitemShut {NoStop}%
	\bibitem [{\citenamefont {Zakharov}\ and\ \citenamefont
		{Shabat}(1973)}]{zakharov}%
	\BibitemOpen
	\bibfield  {author} {\bibinfo {author} {\bibfnamefont {V.}~\bibnamefont
			{Zakharov}}\ and\ \bibinfo {author} {\bibfnamefont {A.}~\bibnamefont
			{Shabat}},\ }\href@noop {} {\bibfield  {journal} {\bibinfo  {journal} {Sov.
				Phys. JETP}\ }\textbf {\bibinfo {volume} {37}},\ \bibinfo {pages} {823}
		(\bibinfo {year} {1973})}\BibitemShut {NoStop}%
	\bibitem [{\citenamefont {Polyakov}(1987)}]{polyakov}%
	\BibitemOpen
	\bibfield  {author} {\bibinfo {author} {\bibfnamefont {A.~M.}\ \bibnamefont
			{Polyakov}},\ }\href@noop {} {\bibfield  {journal} {\bibinfo  {journal} {Mod.
				Phys. Lett. A}\ }\textbf {\bibinfo {volume} {2}},\ \bibinfo {pages} {893}
		(\bibinfo {year} {1987})}\BibitemShut {NoStop}%
	\bibitem [{\citenamefont {Lina}\ and\ \citenamefont {Panigrahi}(1991)}]{lina}%
	\BibitemOpen
	\bibfield  {author} {\bibinfo {author} {\bibfnamefont {J.-M.}\ \bibnamefont
			{Lina}}\ and\ \bibinfo {author} {\bibfnamefont {P.~K.}\ \bibnamefont
			{Panigrahi}},\ }\href@noop {} {\bibfield  {journal} {\bibinfo  {journal}
			{Mod. Phys. Lett. A}\ }\textbf {\bibinfo {volume} {6}},\ \bibinfo {pages}
		{3517} (\bibinfo {year} {1991})}\BibitemShut {NoStop}%
	\bibitem [{\citenamefont {Frantzeskakis}(2010)}]{frantz}%
	\BibitemOpen
	\bibfield  {author} {\bibinfo {author} {\bibfnamefont {D.}~\bibnamefont
			{Frantzeskakis}},\ }\href@noop {} {\bibfield  {journal} {\bibinfo  {journal}
			{J. Phys. A}\ }\textbf {\bibinfo {volume} {43}},\ \bibinfo {pages} {213001}
		(\bibinfo {year} {2010})}\BibitemShut {NoStop}%
	\bibitem [{\citenamefont {Carr}\ \emph {et~al.}(2000)\citenamefont {Carr},
		\citenamefont {Clark},\ and\ \citenamefont {Reinhardt}}]{carr}%
	\BibitemOpen
	\bibfield  {author} {\bibinfo {author} {\bibfnamefont {L.~D.}\ \bibnamefont
			{Carr}}, \bibinfo {author} {\bibfnamefont {C.~W.}\ \bibnamefont {Clark}}, \
		and\ \bibinfo {author} {\bibfnamefont {W.~P.}\ \bibnamefont {Reinhardt}},\
	}\href@noop {} {\bibfield  {journal} {\bibinfo  {journal} {Phys. Rev. A}\
		}\textbf {\bibinfo {volume} {62}},\ \bibinfo {pages} {063610} (\bibinfo
		{year} {2000})}\BibitemShut {NoStop}%
	\bibitem [{\citenamefont {Raju}\ \emph {et~al.}(2005)\citenamefont {Raju},
		\citenamefont {Kumar},\ and\ \citenamefont {Panigrahi}}]{soloman}%
	\BibitemOpen
	\bibfield  {author} {\bibinfo {author} {\bibfnamefont {T.}~\bibnamefont
			{Raju}}, \bibinfo {author} {\bibfnamefont {C.}~\bibnamefont {Kumar}}, \ and\
		\bibinfo {author} {\bibfnamefont {P.}~\bibnamefont {Panigrahi}},\ }\href@noop
	{} {\bibfield  {journal} {\bibinfo  {journal} {J. Phys. A}\ }\textbf
		{\bibinfo {volume} {38}} (\bibinfo {year} {2005})}\BibitemShut {NoStop}%
	\bibitem [{\citenamefont {Vyas}\ \emph {et~al.}(2008)\citenamefont {Vyas},
		\citenamefont {Patel}, \citenamefont {Panigrahi}, \citenamefont {Kumar},\
		and\ \citenamefont {Greiner}}]{vyas}%
	\BibitemOpen
	\bibfield  {author} {\bibinfo {author} {\bibfnamefont {V.~M.}\ \bibnamefont
			{Vyas}}, \bibinfo {author} {\bibfnamefont {P.}~\bibnamefont {Patel}},
		\bibinfo {author} {\bibfnamefont {P.~K.}\ \bibnamefont {Panigrahi}}, \bibinfo
		{author} {\bibfnamefont {C.~N.}\ \bibnamefont {Kumar}}, \ and\ \bibinfo
		{author} {\bibfnamefont {W.}~\bibnamefont {Greiner}},\ }\href@noop {}
	{\bibfield  {journal} {\bibinfo  {journal} {Phys. Rev. A}\ }\textbf {\bibinfo
			{volume} {78}},\ \bibinfo {pages} {021803} (\bibinfo {year}
		{2008})}\BibitemShut {NoStop}%
	\bibitem [{\citenamefont {Alka}\ \emph {et~al.}(2011)\citenamefont {Alka},
		\citenamefont {Goyal}, \citenamefont {Gupta}, \citenamefont {Kumar},\ and\
		\citenamefont {Raju}}]{cnk}%
	\BibitemOpen
	\bibfield  {author} {\bibinfo {author} {\bibnamefont {Alka}}, \bibinfo
		{author} {\bibfnamefont {A.}~\bibnamefont {Goyal}}, \bibinfo {author}
		{\bibfnamefont {R.}~\bibnamefont {Gupta}}, \bibinfo {author} {\bibfnamefont
			{C.~N.}\ \bibnamefont {Kumar}}, \ and\ \bibinfo {author} {\bibfnamefont
			{T.~S.}\ \bibnamefont {Raju}},\ }\href@noop {} {\bibfield  {journal}
		{\bibinfo  {journal} {Phys. Rev. A}\ }\textbf {\bibinfo {volume} {84}},\
		\bibinfo {pages} {063830} (\bibinfo {year} {2011})}\BibitemShut {NoStop}%
	\bibitem [{\citenamefont {Schwinger}(1962)}]{sch}%
	\BibitemOpen
	\bibfield  {author} {\bibinfo {author} {\bibfnamefont {J.}~\bibnamefont
			{Schwinger}},\ }\href@noop {} {\bibfield  {journal} {\bibinfo  {journal}
			{Phys. Rev.}\ }\textbf {\bibinfo {volume} {128}},\ \bibinfo {pages} {2425}
		(\bibinfo {year} {1962})}\BibitemShut {NoStop}%
	\bibitem [{\citenamefont {Das}(1993)}]{das}%
	\BibitemOpen
	\bibfield  {author} {\bibinfo {author} {\bibfnamefont {A.}~\bibnamefont
			{Das}},\ }\href@noop {} {\emph {\bibinfo {title} {Field theory: a path
				integral approach}}},\ Vol.~\bibinfo {volume} {52}\ (\bibinfo  {publisher}
	{World Scientific},\ \bibinfo {year} {1993})\BibitemShut {NoStop}%
	\bibitem [{\citenamefont {Itzykson}\ and\ \citenamefont {Zuber}(2005)}]{iz}%
	\BibitemOpen
	\bibfield  {author} {\bibinfo {author} {\bibfnamefont {C.}~\bibnamefont
			{Itzykson}}\ and\ \bibinfo {author} {\bibfnamefont {J.~B.}\ \bibnamefont
			{Zuber}},\ }\href@noop {} {\emph {\bibinfo {title} {Quantum field theory}}}\
	(\bibinfo  {publisher} {Dover Publications},\ \bibinfo {year}
	{2005})\BibitemShut {NoStop}%
	\bibitem [{\citenamefont {Patani}\ \emph {et~al.}(1976)\citenamefont {Patani},
		\citenamefont {Schlindwein},\ and\ \citenamefont {Shafi}}]{patani}%
	\BibitemOpen
	\bibfield  {author} {\bibinfo {author} {\bibfnamefont {A.}~\bibnamefont
			{Patani}}, \bibinfo {author} {\bibfnamefont {M.}~\bibnamefont {Schlindwein}},
		\ and\ \bibinfo {author} {\bibfnamefont {Q.}~\bibnamefont {Shafi}},\
	}\href@noop {} {\bibfield  {journal} {\bibinfo  {journal} {Journal of Physics
				A: Mathematical and General}\ }\textbf {\bibinfo {volume} {9}},\ \bibinfo
		{pages} {1513} (\bibinfo {year} {1976})}\BibitemShut {NoStop}%
	\bibitem [{\citenamefont {Brading}\ and\ \citenamefont
		{Brown}(2000)}]{brading}%
	\BibitemOpen
	\bibfield  {author} {\bibinfo {author} {\bibfnamefont {K.}~\bibnamefont
			{Brading}}\ and\ \bibinfo {author} {\bibfnamefont {H.}~\bibnamefont
			{Brown}},\ }\href@noop {} {\bibfield  {journal} {\bibinfo  {journal}
			{arXiv:hep-th/}\ }\textbf {\bibinfo {volume} {0009058}} (\bibinfo {year}
		{2000})}\BibitemShut {NoStop}%
	\bibitem [{\citenamefont {Jackiw}\ and\ \citenamefont {Pi}(2015)}]{jack}%
	\BibitemOpen
	\bibfield  {author} {\bibinfo {author} {\bibfnamefont {R.}~\bibnamefont
			{Jackiw}}\ and\ \bibinfo {author} {\bibfnamefont {S.-Y.}\ \bibnamefont
			{Pi}},\ }\href@noop {} {\bibfield  {journal} {\bibinfo  {journal} {Phys. Rev.
				D}\ }\textbf {\bibinfo {volume} {91}},\ \bibinfo {pages} {067501} (\bibinfo
		{year} {2015})}\BibitemShut {NoStop}%
	\bibitem [{\citenamefont {Shukla}\ \emph {et~al.}(2016)\citenamefont {Shukla},
		\citenamefont {Abhinav},\ and\ \citenamefont {Panigrahi}}]{shukla}%
	\BibitemOpen
	\bibfield  {author} {\bibinfo {author} {\bibfnamefont {A.}~\bibnamefont
			{Shukla}}, \bibinfo {author} {\bibfnamefont {K.}~\bibnamefont {Abhinav}}, \
		and\ \bibinfo {author} {\bibfnamefont {P.~K.}\ \bibnamefont {Panigrahi}},\
	}\href@noop {} {\bibfield  {journal} {\bibinfo  {journal} {Classical and
				Quantum Gravity}\ }\textbf {\bibinfo {volume} {33}},\ \bibinfo {pages}
		{235008} (\bibinfo {year} {2016})}\BibitemShut {NoStop}%
	\bibitem [{\citenamefont {Munkres}(1975)}]{munkres}%
	\BibitemOpen
	\bibfield  {author} {\bibinfo {author} {\bibfnamefont {J.~R.}\ \bibnamefont
			{Munkres}},\ }\href@noop {} {\emph {\bibinfo {title} {Topology: a first
				course}}},\ Vol.~\bibinfo {volume} {23}\ (\bibinfo  {publisher}
	{Prentice-Hall Englewood Cliffs, NJ},\ \bibinfo {year} {1975})\BibitemShut
	{NoStop}%
\end{thebibliography}

%

\end{document}